\documentclass[11pt]{article}
\makeatletter
\usepackage[T1]{fontenc}
\usepackage[utf8]{inputenc}
\usepackage[english]{babel}
\usepackage{lmodern}
\usepackage{microtype}
\usepackage{csquotes}
\usepackage[a4paper,top=2.35cm,bottom=2.35cm,left=2.55cm,right=2.55cm,footskip=1.15cm]{geometry}
\usepackage{amsmath,amssymb,amsthm,mathtools}
\usepackage{booktabs,tabularx,longtable,array,multirow}
\newcolumntype{Y}{>{\raggedright\arraybackslash}X}
\newcolumntype{L}[1]{>{\raggedright\arraybackslash}p{#1}}
\usepackage{graphicx}
\usepackage{enumitem}
\usepackage{xcolor}
\definecolor{elpblack}{gray}{0.00}
\definecolor{elpgraydark}{gray}{0.25}
\definecolor{elpgraymid}{gray}{0.55}
\definecolor{elpgraylight}{gray}{0.96}
\usepackage[hyperfootnotes=false]{hyperref}
\usepackage[nameinlink,noabbrev]{cleveref}
\usepackage{titlesec}
\usepackage{fancyhdr}
\usepackage{setspace}
\usepackage{caption}
\usepackage{subcaption}
\usepackage{url}
\usepackage{float}
\usepackage{needspace}
\usepackage{tikz}
\usetikzlibrary{arrows.meta,positioning,calc,fit,backgrounds,decorations.pathreplacing,patterns}
\usepackage[
 backend=biber,style=authoryear,natbib=true,maxcitenames=2,maxbibnames=20,
 uniquename=false,doi=true,url=true,isbn=false,giveninits=true,dashed=false
]{biblatex}
\defbibheading{bibliography}[\refname]{\section*{#1}}
\hypersetup{hidelinks,pdfborder={0 0 0}}
\setlist{nosep,leftmargin=*}
\titleformat{\section}{\Needspace{7\baselineskip}\large\bfseries}{\thesection}{0.65em}{}
\titleformat{\subsection}{\Needspace{5\baselineskip}\normalsize\bfseries}{\thesubsection}{0.65em}{}
\titleformat{\subsubsection}{\normalsize\itshape}{\thesubsubsection}{0.65em}{}
\fancypagestyle{plain}{\fancyhf{}\fancyfoot[C]{\small\thepage}}
\newtheorem{theorem}{Theorem}[section]
\newtheorem{proposition}[theorem]{Proposition}
\newtheorem{lemma}[theorem]{Lemma}

\theoremstyle{definition}

\newtheorem{criterion}[theorem]{Criterion}

\theoremstyle{remark}

\newcommand{\paperversion}[1]{\gdef\PaperVersion{#1}}
\newcommand{\paperdate}[1]{\gdef\PaperDate{#1}}
\newcommand{\paperseries}[1]{\gdef\PaperSeries{#1}}
\newcommand{\papersubtitle}[1]{\gdef\PaperSubtitle{#1}}
\newcommand{\paperaffiliation}[1]{\gdef\PaperAffiliation{#1}}
\newcommand{\papercontact}[1]{\gdef\PaperContact{#1}}

\AtBeginDocument{%
  \providecommand{\PaperVersion}{draft}%
  \providecommand{\PaperDate}{}%
  \providecommand{\PaperSeries}{ForesightFlow Research \textperiodcentered{} Event-Linked Perpetuals}%
  \providecommand{\PaperSubtitle}{}%
  \providecommand{\PaperAffiliation}{}%
  \providecommand{\PaperContact}{}%
}

\tikzset{
 pbox/.style={draw=black,rounded corners=2pt,align=center,inner sep=5pt,font=\footnotesize,text width=23mm,minimum height=9mm,fill=white},
 pwide/.style={draw=black,rounded corners=2pt,align=center,inner sep=5pt,font=\footnotesize,text width=31mm,minimum height=10mm,fill=white},
 psoft/.style={draw=black,dashed,rounded corners=2pt,align=center,inner sep=5pt,font=\footnotesize,text width=26mm,minimum height=9mm,fill=white},
 parr/.style={-{Latex[length=2.0mm]},line width=0.65pt,draw=black},
 pdash/.style={-{Latex[length=2.0mm]},line width=0.65pt,dashed,draw=black},
 pbrace/.style={decorate,decoration={brace,amplitude=4pt},line width=0.55pt,draw=black},
 ffbox/.style={pbox},ffsmall/.style={pbox,font=\scriptsize,inner sep=4pt},ffsoft/.style={psoft},ffarrow/.style={parr},ffdashed/.style={pdash}
}
\makeatletter
\renewcommand{\maketitle}{%
  \begin{center}
    {\small\scshape \PaperSeries\par}
    \vspace{0.8em}
    {\LARGE\bfseries \@title\par}
    \ifx\PaperSubtitle\@empty\else\vspace{0.35em}{\large \PaperSubtitle\par}\fi
    \vspace{0.75em}
    {\large \@author\par}
    \ifx\PaperAffiliation\@empty\else\vspace{0.18em}{\small \PaperAffiliation\par}\fi
    \ifx\PaperContact\@empty\else\vspace{0.12em}{\small \PaperContact\par}\fi
    \vspace{0.32em}{\small \mbox{\PaperDate\space\textbar{}\space Version \PaperVersion}\par}
  \end{center}
  \vspace{0.25em}
}
\makeatother

\makeatother
\hypersetup{pdftitle={Outcome Determination and Settlement Finality on Kalshi: Public State Paths, Prospective Measurement, and Empirical Identification},pdfauthor={Maksym Nechepurenko},pdfsubject={Event-Linked Perpetuals --- Kalshi Research Track},pdfkeywords={Kalshi, prediction markets, finality, settlement, lifecycle, censoring}}
\paperseries{ForesightFlow Research \textperiodcentered{} Event-Linked Perpetuals \textperiodcentered{} Paper 7.1}
\paperversion{r1.3.4}
\paperdate{September 2026}
\papersubtitle{Public State Paths, Prospective Measurement, and Empirical Identification}
\paperaffiliation{Research Department of Devnull FZCO, Dubai, UAE}
\papercontact{\href{mailto:maksym@devnull.ae}{maksym@devnull.ae} \quad ORCID 0000-0002-9515-8841}
\title{Outcome Determination and Settlement Finality on Kalshi}
\author{Maksym Nechepurenko}
\date{}
\begin{document}
\maketitle
\begingroup
\renewcommand{\thefootnote}{}
\footnotetext{\footnotesize\textit{ELP --- Kalshi Research Track.} This paper forms part of the Kalshi-focused research track within the Event-Linked Perpetuals series.}
\endgroup
\begingroup\footnotesize
\begin{abstract}
Event-contract settlement is a state path rather than a universal timestamp. Using a registered seven-day enrollment and seven-day administrative follow-up on Kalshi, this paper distinguishes market-object membership, lifecycle messages, endpoint states, and observation gaps. The MVE layer contains 7,611,594 tickers created inside the registered half-open interval; the reconstructed ordinary stock contains 152,694 markets at risk at the enrollment boundary, 3,477,227 already terminal records, and 11,410 boundary-uncertain records.

The follow-up identifies exact public endpoints for 7,357,576 MVE market objects: 7,237,993 through REST settlement time alone, 119,212 agreeing across WebSocket and REST, and 371 WebSocket-only; 254,018 lack an admissible exact endpoint. In the ordinary boundary cohort, 71,657 markets have an exact WebSocket endpoint and 81,037 do not. These are venue-level public-finality outcomes, not member-cash observations.

Exact determination-to-endpoint pairs are available for 126,806 MVE and 70,979 ordinary tickers, with respective mean durations of 9.676 and 322.50 seconds and observed maxima of 15 and 7,133 seconds. The contrast is descriptive, not causal. Transport gaps preserve two exact endpoint clocks but prevent the same rows from proving complete intermediate paths or absence of revisions; current settlement timers cannot be bound to historical rule versions. Public endpoint completion is therefore extensive, especially on REST, while path completeness remains narrower and family-specific. The exact-key historical crosswalk records earlier interim observations without changing the final cohort or endpoint facts.
\end{abstract}
\endgroup
\noindent\textbf{Keywords:} prediction markets; event contracts; finality; settlement; Kalshi; lifecycle; censoring; observation surfaces.\\
\textbf{JEL:} G13, G14, G18.

\section{Introduction}
Prediction-market research usually examines calibration, information aggregation, or market design \parencite{wolfers2004prediction,manski2006interpreting}. This paper studies a later boundary: the institutional path after ordinary trading approaches or reaches its end. A contract can close, receive a first result, remain contestable, be re-determined, and finalize only later. These intervals matter for capital lock, operational exposure, and any leveraged product that references an event contract.

Kalshi's public architecture makes the path measurable but not self-assembling. The documented state vocabulary includes initialized, active, inactive, closed, determined, disputed, amended, and finalized states. Some changes are explicit lifecycle messages; others are time-driven or visible only in changing representations. The public closed-state filter also coarsens several nonfinal states, and event-level status can depend on child markets \parencite{kalshiMarketLifecycle2026,kalshiLifecycleWS2026}. A terminal object therefore cannot be read as a complete history.

The empirical design separates four objects that are often conflated. First, a REST creation-time population defines market-object membership. Second, lifecycle streams reveal selected state changes and exact payload clocks. Third, current or terminal snapshots supply endpoint fields that can be exact without supplying the intervening path. Fourth, session gaps determine which path-prevalence and revision claims remain identifiable. A result can therefore have an exact finalization timestamp and still lack a complete observed path.

The paper asks: which finality transitions are exactly or interval-observable; how many enrolled markets reach an admissible endpoint by a fixed follow-up boundary; how long exact determination-to-endpoint intervals last; how much endpoint evidence is visible only in REST or WebSocket; and which stronger path, timer, revision, and member-cash claims remain blocked by the public observation process.

The contribution is empirical and methodological. The registered 7+7 design now closes the cohort and endpoint layers at exact ticker level, while preserving the difference between exact two-clock identification and complete-path observation. The resulting evidence shows that public finality can be nearly complete on a terminal REST surface even when only a much smaller subset has a fully linked determination path. Settlement speed is therefore not one scalar until the endpoint, risk set, observation surface, and path-completeness requirement are specified.

\section{Institutional object and evidence surfaces}
Let the public state space be
\[
\mathcal S_K=\{\text{initialized, active, inactive, closed, determined, disputed, amended, finalized}\}.
\]
The empirical object is the market-level valid-time process $S_i(t)$, not merely the representation returned at retrieval time. For retrieval $u$ write
\[
R_i(u)=\{\widehat S_i(u),T_i(u),\Gamma_i(u),M_i(u)\},
\]
where $T_i$ contains timing fields, $\Gamma_i$ the applicable rule/source representation, and $M_i$ other metadata. Lifecycle messages form a separate stream $E_i$, while content-addressed rule and source versions form $D_i$. The reconstructed history is
\[
H_i=\{R_i(u_k)\}_{k=1}^{K_i}\cup E_i\cup D_i.
\]

The finality vector is
\[
F_i^K=(t_{occ},t_{source},t_{decide},t_{close},t_{det,1},t_{det,J},t_{final},t_{member}),
\]
with each coordinate labelled exact, interval-observed, snapshot-only, proxy, unmeasured, or conflicting. Paper 7.1 owns the venue clocks $t_{close}$, $t_{det,1}$, $t_{det,J}$, and $t_{final}$; Papers 7.3 and 7.4 consume them without redefining them \parencite{elp73preprint,elp74preprint}.

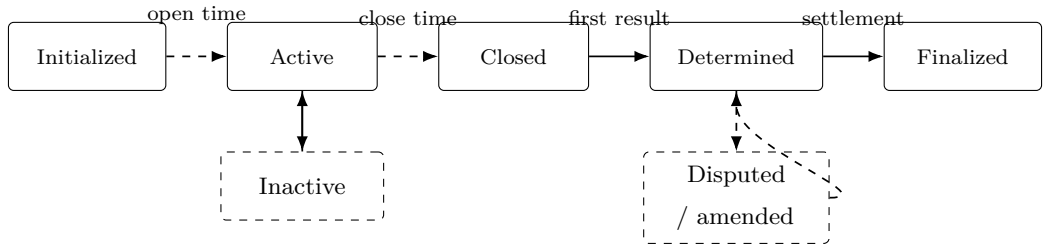
\begin{figure}[H]
\centering
\begin{tikzpicture}[node distance=7mm and 8mm]
\node[ffsmall,text width=18mm] (init) {Initialized};
\node[ffsmall,text width=17mm,right=of init] (active) {Active};
\node[ffsmall,text width=17mm,right=of active] (closed) {Closed};
\node[ffsmall,text width=20mm,right=of closed] (det) {Determined};
\node[ffsmall,text width=18mm,right=of det] (final) {Finalized};
\node[ffsoft,text width=18mm,below=8mm of active] (inactive) {Inactive};
\node[ffsoft,text width=21mm,below=8mm of det] (contest) {Disputed / amended};
\draw[ffdashed] (init.east)--node[midway,above=8pt,font=\scriptsize]{open time}(active.west);
\draw[ffdashed] (active.east)--node[midway,above=8pt,font=\scriptsize]{close time}(closed.west);
\draw[ffarrow] (closed.east)--node[midway,above=8pt,font=\scriptsize]{first result}(det.west);
\draw[ffarrow] (det.east)--node[midway,above=8pt,font=\scriptsize]{settlement}(final.west);
\draw[ffarrow] (active)--(inactive);
\draw[ffarrow] (inactive)--(active);
\draw[ffdashed] (det)--(contest);
\draw[ffdashed] (contest.east) to[out=0,in=-90] (det.south);
\end{tikzpicture}
\caption{Measurement state graph. Solid arrows denote transitions with dedicated public lifecycle evidence in the frozen documentation; dashed arrows denote time-driven or incompletely messaged paths that require reconciliation.}
\label{fig:71-state}
\end{figure}

This is public institutional finality, not legal finality of the entire clearing system. The distinction follows the broader emphasis on explicit rules and settlement risk in financial-market infrastructure while keeping the empirical claim narrower \parencite{cpssiosco2012pfmi}.

\section{Identification}
\begin{proposition}[Terminal projection is non-injective]
A final REST representation containing result, current fields, and settlement timestamp does not identify the complete lifecycle path.
\end{proposition}
\emph{Proof sketch.} A direct close--determine--finalize path and a path containing deactivation, reopening, repeated determination, or overwritten close time can produce the same terminal projection. Any path statistic based only on that projection is therefore not point-identified.

\begin{lemma}[Polling identifies a transition interval]
If snapshots at $u_0<u_1$ observe $S_i(u_0)=z$ and $S_i(u_1)=z'\ne z$, and no matching exact lifecycle event is available, then $t_{z\rightarrow z'}\in(u_0,u_1]$.
\end{lemma}
Assigning $u_1$ as an exact timestamp would add an unregistered measurement model. The ledger preserves the interval.

\begin{proposition}[Prospective inception paths do not require historical exhaustion]
Let $\mathcal C[T_0,T_1)$ be markets whose creation messages are completely captured during a registered interval. Conditional on audited channel continuity, their post-creation paths are identifiable from prospective evidence even when the pre-$T_0$ historical universe is incomplete.
\end{proposition}
The historical archive affects claims about the stock of earlier markets, not membership in a prospectively observed birth cohort.

\begin{criterion}[Path admissibility]
A market enters a full-path denominator only when its coverage class establishes observation from the first transition required by the claim through the endpoint or censoring time. Terminal records may support terminal-field distributions but not unobserved path prevalence.
\end{criterion}

\section{Prospective design}
\subsection{Boundary-reconciled start population}
The ordinary lifecycle collector is authenticated and subscribed before $T_0$. After subscription acknowledgement, $T_0$ is fixed and the live market route is paginated with MVE excluded \parencite{kalshiGetMarkets2026,kalshiLifecycleWS2026}. Because pagination takes time, the baseline is not treated as an instantaneous photograph. Each page carries retrieval time; lifecycle messages received during the crawl reconcile creations and state changes. Markets created after $T_0$ belong to the inception cohort even if encountered by REST. Markets present before $T_0$ and nonterminal at the boundary form the at-risk stock.

\subsection{Enrollment and follow-up}
Every ordinary or MVE creation message observed during $[T_0,T_0+7d)$ enters the corresponding inception cohort. Ordinary and multivariate lifecycle streams remain separate \parencite{kalshiLifecycleWS2026,kalshiMVELifecycle2026}. All enrolled markets are followed through $T_0+14d$; new markets during days 8--14 remain in raw capture but not in the primary cohort.

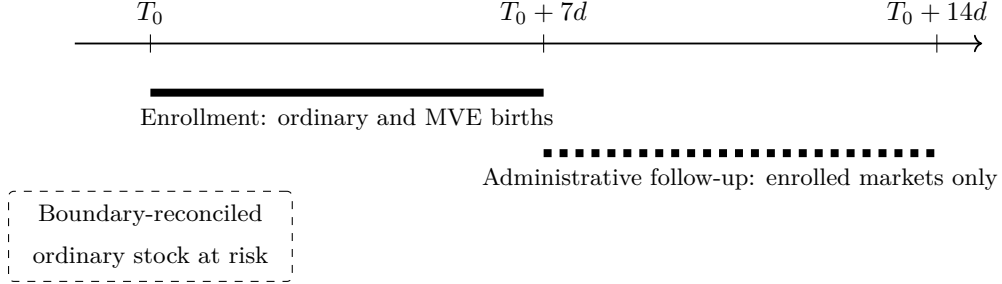
\begin{figure}[H]
\centering
\begin{tikzpicture}[x=1cm,y=1cm]
\draw[->,line width=.7pt] (.3,0)--(12.3,0);
\foreach \x/\lab in {1.3/{$T_0$},6.5/{$T_0+7d$},11.7/{$T_0+14d$}}{\draw (\x,-.12)--(\x,.12);\node[above,font=\small] at (\x,.12){\lab};}
\draw[line width=3pt] (1.3,-.65)--(6.5,-.65);
\node[below,font=\footnotesize,align=center] at (3.9,-.72){Enrollment: ordinary and MVE births};
\draw[line width=3pt,dashed] (6.5,-1.45)--(11.7,-1.45);
\node[below,font=\footnotesize,align=center] at (9.1,-1.52){Administrative follow-up: enrolled markets only};
\node[ffsoft,text width=34mm] at (1.3,-2.55){Boundary-reconciled ordinary stock at risk};
\end{tikzpicture}
\caption{Registered prospective design. The stock cohort is boundary-reconciled; inception cohorts are defined by prospective creation evidence.}
\label{fig:71-cohort}
\end{figure}

\subsection{Reconciliation and coverage classes}
REST refreshes occur after material lifecycle events, around close and expected-expiration boundaries, after determination or re-determination, after settlement, and after reconnect gaps. The analysis distinguishes:
\begin{itemize}
\item complete prospective observation from creation;
\item prospective observation of the stock at the enrollment boundary;
\item left-truncated prospective observation;
\item historical terminal observation only;
\item interval-reconstructed observation; and
\item conflicting evidence.
\end{itemize}
The bounded current cross-sectional subset is used for schema, source/rule, and terminal cross-sectional distributions, not for historical path rates.

\section{Enrollment population and observation-surface results}\label{sec:71-interim-results}

The population layer is bound to the registered half-open interval
\[
[T_0,T_{\mathrm{enroll}})=
[\text{2026-08-15 19:51:29.025 UTC},
  \text{2026-08-22 19:51:29.025 UTC}).
\]
The initial monolithic cursor chain was excluded after an out-of-scope-row and cursor-link audit. Population authority instead comes from 190 terminal-validated leaf shards whose logical union covers all 604,800 seconds. Durable per-leaf outputs and a global union reproduce
\[
7{,}617{,}371
=
7{,}611{,}594 + 5{,}777,
\]
where the first term is the raw reference count, the second is the locally included unique-ticker population, and the third is the retained boundary-overlap exclusion count \parencite{nechepurenko2026kmve}. Exact duplicate groups and parent/identity conflicts are zero in the global union.

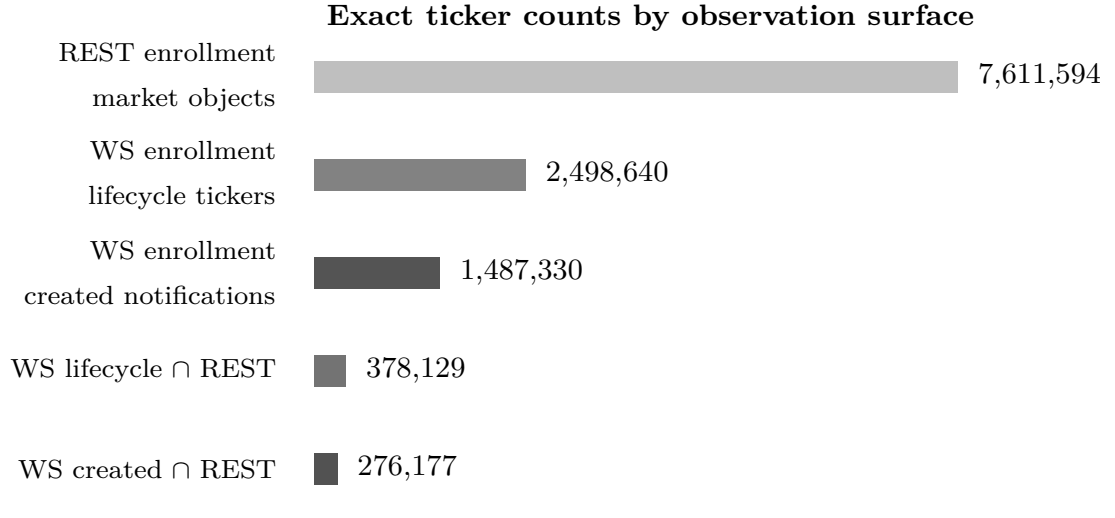
\begin{table}[H]
\centering\footnotesize
\caption{Frozen Kalshi enrollment-population and observation-surface denominators. Counts have different statistical units and are not an additive funnel.}
\begin{tabularx}{\textwidth}{@{}XrX@{}}
\toprule
Object & Exact ticker count & Interpretation \\
\midrule
REST enrollment live materialization & 7,611,594 & Market objects with REST creation time inside the registered interval \\
WS enrollment lifecycle & 2,498,640 & Exact tickers with any qualifying lifecycle observation during enrollment \\
WS enrollment created notifications & 1,487,330 & Exact tickers with an observed creation notification during enrollment \\
WS left-truncated evidence class & 1,011,310 & Lifecycle tickers without an enrollment creation notification \\
WS post-enrollment continuation & 107,914 & Post-boundary tickers retained outside the primary cohort \\
WS created $\cap$ REST & 276,177 & Exact ticker intersection, without gap extrapolation \\
WS lifecycle $\cap$ REST & 378,129 & Exact ticker intersection, without gap extrapolation \\
\bottomrule
\end{tabularx}
\label{tab:71-denominators}
\end{table}

\begin{figure}[H]
\centering
\resizebox{\textwidth}{!}{\begin{tikzpicture}[x=1cm,y=.84cm,font=\footnotesize]
\node[anchor=west,text=elpblack] at (0,6.15) {\textbf{Exact ticker counts by observation surface}};
\node[anchor=east,align=right,text width=3.65cm,font=\scriptsize] at (-.28,5.40) {REST enrollment\\market objects};
\fill[black!25] (0,5.20) rectangle (6.9000,5.60);
\node[anchor=west,text=elpblack] at (6.9800,5.40) {7,611,594};
\node[anchor=east,align=right,text width=3.65cm,font=\scriptsize] at (-.28,4.15) {WS enrollment\\lifecycle tickers};
\fill[black!49] (0,3.95) rectangle (2.2650,4.35);
\node[anchor=west,text=elpblack] at (2.3450,4.15) {2,498,640};
\node[anchor=east,align=right,text width=3.65cm,font=\scriptsize] at (-.28,2.90) {WS enrollment\\created notifications};
\fill[black!67] (0,2.70) rectangle (1.3483,3.10);
\node[anchor=west,text=elpblack] at (1.4283,2.90) {1,487,330};
\node[anchor=east,align=right,text width=3.65cm,font=\scriptsize] at (-.28,1.65) {WS lifecycle $\cap$ REST};
\fill[black!55] (0,1.45) rectangle (0.3428,1.85);
\node[anchor=west,text=elpblack] at (0.4228,1.65) {378,129};
\node[anchor=east,align=right,text width=3.65cm,font=\scriptsize] at (-.28,.40) {WS created $\cap$ REST};
\fill[black!68] (0,.20) rectangle (0.2504,.60);
\node[anchor=west,text=elpblack] at (0.3304,.40) {276,177};
\draw[black!35] (0,-.24) -- (6.9,-.24);
\node[anchor=north west,text=black!75,align=left,text width=8.2cm,font=\scriptsize] at (-3.8,-.48) {Distinct observation surfaces; bar lengths do not form a common population funnel.};
\end{tikzpicture}
}
\caption{Ticker counts on the frozen REST and WebSocket observation surfaces. Bar height is not evidence of one common population: REST membership is defined by current market records and exact creation time, while the WebSocket sets are notification histories under documented session gaps.}
\label{fig:71-observation-surfaces}
\end{figure}

The created-notification intersection represents 18.57\% of the WS created-notification set and 3.63\% of the REST creation-time population. The all-lifecycle intersection represents 15.13\% of the WS lifecycle set and 4.97\% of the REST population. These ratios diagnose surface non-equivalence; they are not capture-recall estimates because the surfaces have different institutional semantics and the WebSocket stream contains left-truncated and transport-gap exposure.

\subsection{What the population release identifies and how the follow-up extends it}
The KMVE V1 release identifies the complete REST MVE market-object population for the enrollment interval, exact boundary treatment, deduplication, field coverage, and exact ticker-set relations across the frozen REST and WebSocket surfaces \parencite{nechepurenko2026kmve}. The subsequent key-level 7+7 follow-up does not redefine that population. It adds three distinct objects:
\begin{enumerate}
  \item a corrected ordinary boundary stock, separating reconstructed-at-risk, terminal-before-boundary, and boundary-uncertain records;
  \item an exact endpoint union across WebSocket and REST evidence through the administrative boundary; and
  \item exact determination-to-endpoint pairs with separate path and gap qualifications.
\end{enumerate}
The follow-up therefore closes public endpoint status for a large share of the frozen cohorts without upgrading every endpoint to a complete lifecycle path. Current fields observed later remain current observations unless they contain an exact historical timestamp. Member-level cash realization remains outside the public-market analysis.

\subsection{Acquisition failures as part of the measurement result}
The collection design uses independently validated temporal shards rather than a single retrieval chain. It shows why population validity depends on explicit boundaries, cursor non-cycles, durable raw evidence, and reproducible derived layers. The resulting population is supported by per-leaf records and a deterministic global union, so the reported counts can be regenerated without relying on a single mutable index.

\section{Final 7+7 follow-up results}\label{sec:final-results}

\subsection{Corrected cohort construction}
The ordinary baseline contained 3,641,331 exact tickers, but it was acquired over time and cannot be treated as one instantaneous at-risk stock. Exact boundary reconstruction classifies 3,477,227 as terminal at or before $T_0$, 152,694 as reconstructed at risk at $T_0$, and 11,410 as boundary-uncertain. The primary ordinary follow-up cohort is the 152,694 reconstructed-at-risk keyset. The archived payload surface does not provide an exact valid-time ordinary birth keyset during enrollment: 140,900 WebSocket-observed birth notifications remain a separate observational class and are not promoted into the exact primary cohort.

The MVE primary cohort remains the archived REST creation-time population of 7,611,594 unique market objects. These two cohorts have different formation rules and cannot be pooled into one venue-wide finality rate.

\begin{table}[H]
\centering\footnotesize
\caption{Corrected 7+7 primary cohorts and boundary classes. Counts are exact tickers.}
\label{tab:cohorts-final}
\begin{tabularx}{\textwidth}{@{}L{0.32\textwidth}rY@{}}
\toprule
Class & Tickers & Interpretation \\
\midrule
Ordinary terminal at/before $T_0$ & 3,477,227 & Outside the prospective at-risk stock \\
Ordinary reconstructed at risk at $T_0$ & 152,694 & Primary ordinary follow-up cohort \\
Ordinary boundary-uncertain & 11,410 & Excluded from primary rates \\
Exact ordinary enrollment births & 0 & Exact valid-time birth keyset not identified in the archived payload \\
MVE REST enrollment objects & 7,611,594 & Exact creation-time membership in $[T_0,T_{enroll})$ \\
\bottomrule
\end{tabularx}
\end{table}

\subsection{Endpoint-source decomposition}
For MVE, 7,357,576 market objects have an exact public endpoint by the follow-up boundary. Most are visible only through an exact REST settlement timestamp; a smaller intersection has both WebSocket and REST evidence, and 371 are WebSocket-only. The remaining 254,018 have no admissible exact endpoint by the boundary. Thus 96.66 percent of the MVE cohort has an exact public finalization endpoint, while 3.34 percent remains without one under the registered evidence policy.

For the ordinary $T_0$ stock, 71,657 markets have an exact WebSocket endpoint and 81,037 do not, corresponding to 46.93 and 53.07 percent of the corrected cohort. These proportions describe the reconstructed boundary stock, not the entire 3.64-million-row ordinary enumeration and not a complete ordinary inception population.

\begin{table}[H]
\centering\footnotesize
\caption{Exact endpoint classes at the 14-day follow-up boundary. Public venue finalization is distinct from member cash realization.}
\label{tab:endpoints-final}
\begin{tabularx}{\textwidth}{@{}L{0.16\textwidth}L{0.31\textwidth}r r Y@{}}
\toprule
Cohort & Endpoint class & Tickers & Cohort share & Evidence interpretation \\
\midrule
MVE & REST exact only & 7,237,993 & 95.09\% & Exact REST settlement time; no exact WS endpoint in the joined registry \\
MVE & WS and REST exact agree & 119,212 & 1.57\% & Same exact endpoint supported on both surfaces \\
MVE & WS exact only & 371 & 0.005\% & Exact lifecycle endpoint without REST endpoint in the joined evidence \\
MVE & No exact endpoint by follow-up & 254,018 & 3.34\% & No admissible exact endpoint; current status is not back-projected \\
\addlinespace
Ordinary & WS exact endpoint & 71,657 & 46.93\% & Exact public endpoint in the reconstructed $T_0$ stock \\
Ordinary & No exact endpoint by follow-up & 81,037 & 53.07\% & No admissible exact endpoint in the registered follow-up \\
\bottomrule
\end{tabularx}
\end{table}

The decomposition is itself a measurement result. A terminal REST surface can identify finalization for millions of market objects while providing no determination path. Conversely, WebSocket evidence supplies transition histories for a much smaller keyset. Neither surface is a universal gold standard.

\subsection{Exact determination-to-endpoint clocks}
The key-level lifecycle registries contain 126,806 exact MVE determination-to-endpoint pairs and 70,979 exact ordinary pairs. The MVE mean is 9.676 seconds and the observed maximum is 15 seconds. The ordinary mean is 322.50 seconds and the maximum is 7,133 seconds. The difference between 71,657 ordinary exact endpoints and 70,979 exact pairs leaves 678 endpoint-observed markets without a compatible exact determination-to-endpoint pair.

\begin{table}[H]
\centering\footnotesize
\caption{Exact public determination-to-endpoint clock subsets. These are ticker-level empirical distributions, not causal venue comparisons.}
\label{tab:durations-final}
\begin{tabularx}{\textwidth}{@{}L{0.18\textwidth}r r r Y@{}}
\toprule
Family & Exact pairs & Mean (s) & Maximum (s) & Qualification \\
\midrule
MVE & 126,806 & 9.676 & 15 & Exact payload/endpoint clocks; collection and event dependence retained \\
Ordinary & 70,979 & 322.50 & 7,133 & Exact clocks in the reconstructed $T_0$ stock; not a full ordinary birth census \\
\bottomrule
\end{tabularx}
\end{table}

The large difference between the family-specific empirical distributions is not a causal MVE-versus-ordinary effect. MVE objects are generated under a concentrated collection architecture, share parent events and legs, and use a different public observation surface. The estimand is native finality within each family.

\subsection{Exact clocks versus path completeness}
The follow-up pairs 94 ordinary-lifecycle-stream and 293 MVE-lifecycle-stream gap/resumption receipts. Endpoint-aware market classification identifies a gap intersecting the relevant observed path for 6,421,306 MVE objects and no such intersection for 1,190,288. In the ordinary stock, only eight markets are gap-intersected and 152,686 are not.

A transport gap does not invalidate two exact payload clocks. It does block stronger claims that no intermediate determination, amendment, or reopening occurred. Accordingly, the exact duration distributions in \cref{tab:durations-final} are admitted as two-clock facts, while population-wide revision incidence and complete-path prevalence remain qualified. The MVE evidence makes the distinction especially visible: 7,230,770 market objects have exact finalization without a linked exact WebSocket lifecycle path.

\subsection{Timer and revision boundaries}
The current REST settlement-timer field is present for all 7,611,594 MVE objects and absent from the ordinary registry. The public evidence does not bind that current field to the historical rule version applicable at the last determination. A primary timer-residual estimand is therefore blocked; only a separately labelled current-timer sensitivity is possible.

No multi-version determination path is identified in the corrected exact-pair registries. Because gap exposure is material for MVE, this is an observed lower bound rather than evidence that revisions never occurred. Absence of an observed dispute or amendment likewise does not establish absence of contestability.

\section{Historical-prefix exact-clock validation}\label{sec:historical-prefix-validation}

The final-cohort estimates above are constructed from the registered seven-day enrollment population and the completed follow-up registry.  A separate reconstruction was therefore used to test whether the earlier exact-clock audit could be reproduced at ticker level without allowing that historical prefix to redefine the final cohort.  The reconstructed registry is bound to the original archived prefixes---895 ordinary-lifecycle-stream chunks (98,554 frames) and 24,065 MVE-lifecycle-stream chunks (2,743,982 frames)---and retains two distinct published pair definitions: first determination to settlement and last determination to settlement.

For each definition, the ordinary lifecycle stream reproduces 31,828 exact positive-duration pairs.  The MVE lifecycle stream reproduces 653,479 exact pairs, of which 130,973 have positive duration and 522,506 have equal determination and settlement timestamps at the documented one-second precision.  The equal-second class is a representation of timestamp precision only.  It does not identify simultaneous internal processing, absence of contestability, member cash availability, or zero capital lock.

\begin{table}[t]
\centering
\small
\caption{Historical-prefix exact-clock registry and exact overlap with the final key-level follow-up registry. Counts are reported separately for each published pair definition; the two definitions have identical aggregate class counts but remain distinct registry records.}
\label{tab:r23a-crosswalk}
\begin{tabular}{@{}lrrrr@{}}
\toprule
Population & Historical pairs & Exact overlap & Historical-only & Final-registry-only \\
\midrule
Ordinary & 31,828 & 29,190 & 2,638 & 41,789 \\
MVE & 653,479 & 21,950 & 631,529 & 104,856 \\
\bottomrule
\end{tabular}
\par\smallskip
\begin{minipage}{0.94\textwidth}\footnotesize\raggedright
\textit{Interpretation.} (1) Ordinary: high overlap; the final registry also contains later or differently qualified pairs. (2) MVE: the historical prefix and final enrollment registry are distinct observation populations.
\end{minipage}
\end{table}

Table~\ref{tab:r23a-crosswalk} reports the bidirectional exact-key comparison with the final follow-up registry.  Under each pair definition, 29,190 ordinary historical-prefix pairs and 21,950 MVE historical-prefix pairs match the final registry exactly.  These correspond to 91.71\% and 3.36\% of the respective historical-prefix registries.  Viewed in the reverse direction, the historical prefix covers 41.12\% of final ordinary pairs and 17.31\% of final MVE pairs.  The residual classes are preserved rather than assigned an economic cause by assumption.  In particular, they are not automatically labelled cohort attrition, missed lifecycle events, or schema loss.

This validation has two consequences.  First, it reproduces the original clock parser and both pair semantics from immutable raw evidence.  Second, it shows why the historical prefix cannot serve as the denominator for the prospective article: especially for MVE markets, it represents a different observation population.  All headline cohort, endpoint, and duration results in this paper therefore continue to use the final key-level follow-up registry.  The historical-prefix registry is used only as a reproducibility and sensitivity layer.

\section{Estimands and inference}
For admissible clocks define
\[
L_i^{CD}=t_{det,1,i}-t_{close,i},\qquad
L_i^{DF}=t_{final,i}-t_{det,1,i},\qquad
L_i^{CF}=t_{final,i}-t_{close,i}.
\]
Let $t_{det,J,i}$ and $\delta_{J,i}$ be the last observed determination and its version-consistent timer. Define the post-timer residual
\[
W_i=t_{final,i}-(t_{det,J,i}+\delta_{J,i})
\]
and the revision interval
\[
R_i=t_{det,J,i}-t_{det,1,i}.
\]
$W_i$ is descriptive; a positive value is not called a rule violation without an applicable requirement and sufficient timestamp precision.

\begin{table}[H]
\centering\footnotesize
\caption{Primary estimand registry.}
\begin{tabularx}{\textwidth}{@{}p{0.20\textwidth}p{0.19\textwidth}p{0.20\textwidth}X@{}}
\toprule
Estimand & Risk-set entry & Endpoint/censoring & Dependence unit \\
\midrule
$L^{CD}$ & reliable close & first determination; administrative censoring & event ticker; series sensitivity \\
$L^{DF}$ & first determination & finalization; administrative censoring & event ticker \\
$W$ & last determination & finalization relative to timer & event ticker \\
$R$ & first determination & last determination & event ticker \\
Path recovery & cohort entry & last reliable observation & acquisition session and event \\
\bottomrule
\end{tabularx}
\end{table}

Exact and interval clocks are reported separately. Single-endpoint survival summaries use Kaplan--Meier only when the censoring formulation is defensible \parencite{kaplan1958nonparametric}. Multistate or competing-transition descriptions use Aalen--Johansen/counting-process methods only for observed transitions \parencite{aalen1978transition,andersen1993counting}. Uncertainty is obtained by cluster bootstrap at event level, with series-level clustering as a sensitivity; market tickers sharing an event are not treated as independent replications.

\section{Registered analyses and interpretation}
\textbf{A1 --- Endpoint-source value.} Compare exact endpoint identification under REST only, WebSocket only, and their exact-ticker union. The primary result is the change in admissible endpoint counts and evidence grades, not a declaration that one channel is correct and the other erroneous.

\textbf{A2 --- Native finality timing.} Report exact determination-to-endpoint distributions separately for ordinary and MVE families. The paper does not infer a causal architecture effect from their raw difference.

\textbf{A3 --- Path completeness.} Distinguish exact two-clock identification from gap-free path and revision-history identification. An intervening transport gap does not erase exact endpoints but blocks an absence-of-revision claim.

\textbf{A4 --- Timer alignment.} Evaluate a historical timer residual only when the rule/timer version applicable to the last determination is identified. Current-timer sensitivity remains descriptive and non-historical.

\textbf{A5 --- Historical-prefix crosswalk.} The earlier interim first/last pair registries are crosswalked to the final cohorts without redefining the final estimands. The analysis records the exact disposition of the earlier equal-second class rather than using it as a substitute denominator.

\section{Validity and reproducibility}
The claims require: durable cohort keysets; exact endpoint-source joins; semantic duplicate control; explicit gap/session lineage; coordinate-specific evidence grades; and reproducible generation of every table and figure. The population and endpoint facts resolve to archived key-level registries. A failed path, timer, or reconciliation check suppresses only the stronger dependent claim rather than replacing it with a convenient proxy.

The interim-to-final exact-key crosswalk audits historical interim counts without changing the registered cohort or endpoint definitions.

\section{Limitations}
The study measures public venue states, not internal deliberation, legal finality, or member-level cash availability. Seven days of enrollment do not represent every calendar regime. The ordinary primary cohort is a reconstructed $T_0$ stock because exact valid-time enrollment births are not identified in the archived payload. The MVE cohort is population-complete at the REST market-object level but highly dependent across events, collections, and selected legs.

REST settlement timestamps identify public finalization without reconstructing the intervening path. WebSocket gaps can intersect path-risk intervals even when both endpoint clocks are exact. Revision and contestability prevalence are therefore lower-bound or blocked claims on affected paths. Current timer fields cannot be projected backward to historical determination versions. The earlier interim pair registry is retained as an exact-key historical-prefix crosswalk only; it does not promote residual classes to a complete historical reproduction.

\section{Conclusion}
Kalshi's public finality is now empirically identified at two distinct layers. The first is the endpoint layer. Exact public finalization is observed for 7,357,576 of 7,611,594 MVE market objects and for 71,657 of 152,694 ordinary markets reconstructed as at risk at the enrollment boundary. The second is the path layer. Only 126,806 MVE and 70,979 ordinary tickers have exact determination-to-endpoint pairs, and path completeness remains narrower still under session gaps.

This divergence is the paper's central result. A venue can expose a nearly complete terminal REST population while revealing exact lifecycle paths for a small subset. MVE exact pairs finalize within a 15-second observed range and average 9.676 seconds; ordinary exact pairs average 322.50 seconds and extend to 7,133 seconds. These are native, family-specific clock distributions, not a causal comparison and not evidence of member cash realization.

The result replaces a scalar notion of settlement speed with a graded finality object: cohort membership, endpoint source, exact clocks, path coverage, timer version, and economic discharge state. The historical interim-prefix reconstruction retains the interim-to-final audit without altering the qualified empirical conclusions supported by the final key-level follow-up.

\appendix
\section{Evidence grades}
Each clock or transition is classified as:
\begin{itemize}
\item exact from an explicit event or source timestamp;
\item interval-observed through polling or current-snapshot only; or
\item secondary proxy, unmeasured, or conflicting evidence.
\end{itemize}
Analytical tables never collapse these classes into one generic nonmissing indicator.

\section*{Data and Code Availability}
The seven-day MVE population and structural data are available as the \emph{Kalshi Multivariate Event Market Materialization Dataset (KMVE)}, Mendeley Data, V1, doi: \href{https://doi.org/10.17632/fn65786cg6.1}{10.17632/fn65786cg6.1}. The lifecycle, gap, and key-level endpoint materials support the reported analyses; restricted operational records and raw authenticated payloads are not distributed with this article. The public materials therefore provide the documentation and reproducible sources appropriate to the evidence described in the manuscript.

\section*{Generative AI Disclosure}
OpenAI ChatGPT and Codex were used for editorial and technical assistance during manuscript preparation. The author made all substantive research decisions, reviewed the final manuscript, and assumes full responsibility for its contents.

\section*{Funding}
This research received no external funding.

\section*{Competing Interests}
The author is affiliated with the Research Department of Devnull FZCO and leads the ForesightFlow research programme. No external sponsor influenced the research design, analysis, interpretation, or decision to publish. The article does not evaluate a commercial product or make investment recommendations.

\printbibliography[heading=bibliography,title={References}]
\end{document}